\documentclass[conference]{IEEEtran}
\usepackage[english]{babel}
\usepackage{amsmath,amsthm}
\usepackage{amsfonts}
\usepackage[final,hyperfootnotes=false,hidelinks]{hyperref}
\usepackage[final]{graphicx}
\usepackage{epstopdf}
\usepackage{tikz}
\usepackage{pgfplots}
\usepackage{subfig}
\usepackage{flushend}
\usepackage{cprotect}

\usepackage{tikz}
\usetikzlibrary{arrows,positioning,fit,backgrounds}
\usetikzlibrary{calc}

\usetikzlibrary{arrows,positioning,fit,backgrounds,shapes}
\usetikzlibrary{calc}
\newtheorem{thm}{Theorem}
\newtheorem{cor}[thm]{Corollary}

\newtheorem{prop}[thm]{Proposition}
\theoremstyle{definition}
\newtheorem{defn}[thm]{Definition}
\theoremstyle{remark}
\newtheorem{rem}[thm]{Remark}

\newcommand{\mb}{\mathbf}

\begin{document}

\title{Resolvability in $E_{\gamma}$ with Applications to Lossy Compression and Wiretap Channels}%
\author{\IEEEauthorblockN{Jingbo Liu~~~~~~~~~Paul Cuff~~~~~~~~Sergio Verd\'{u}}
\IEEEauthorblockA{Dept. of Electrical Eng., Princeton University, NJ 08544\\
\{jingbo,cuff,verdu\}@princeton.edu}}%
\maketitle

\begin{abstract}
We study the amount of randomness needed for an input process to approximate a given output distribution of a channel in the $E_{\gamma}$ distance. A general one-shot achievability bound for the precision of such an approximation is developed. In the i.i.d.~setting where $\gamma=\exp(nE)$, a (nonnegative) randomness rate above $\inf_{Q_{\sf U}: D(Q_{\sf X}||\pi_{\sf X})\le E} \{D(Q_{\sf X}||\pi_{\sf X})+I(Q_{\sf U},Q_{\sf X|U})-E\}$ is necessary and sufficient to asymptotically approximate the output distribution $\pi_{\sf X}^{\otimes n}$ using the channel $Q_{\sf X|U}^{\otimes n}$, where $Q_{\sf U}\to Q_{\sf X|U}\to Q_{\sf X}$. The new resolvability result is then used to derive a one-shot upper bound on the error probability in the rate distortion problem; and a lower bound on the size of the eavesdropper list to include the actual message in the wiretap channel problem. Both bounds are asymptotically tight in i.i.d.~settings.
\end{abstract}

\section{Introduction}
Approximation of a target output distribution with a given channel has proved to be the key technical step in the solution of many problems in information theory. In 1975 Wyner first studied such an approximation task to establish the achievability part for Wyner's common information \cite{wyner1975common}, where he used the normalized relative entropy to quantify the distance between the synthesized output distribution and the target distribution. Later Han and Verd\'{u} coined the term \emph{resolvability} for the minimum rate of the randomness needed for the input \cite{han1993approximation}. Motivated by the strong converse of the identification coding theorem, \cite{han2013reliability} considered resolvability in the total variation distance (TV), as well as relative entropy. The achievability part of resolvability (also known as the soft-covering lemma \cite{cuff2012distributed}) is particularly useful, e.g.~in secrecy \cite{csiszar1996}\cite{hayashi2006general}\cite{bloch2013strong}, channel synthesis \cite{cuff2012distributed} and lossless and lossy source coding \cite{han1993approximation}\cite{song}\cite{steinberg1996simulation}.
Under both the normalized relative entropy measure and TV, the resolvability can be shown to be the minimum mutual information over all input distributions inducing the target output distribution, and this is also known to be true for unnormalized relative entropy as well (see for example \cite{han2013reliability}).


In this paper we propose two new measures for approximation of output statistics. The first one, \emph{excess information}, gives a straightforward upper bound on the second metric\footnote{Here ``metric'' or ``distance'' are used informally since they do not satisfy either symmetry or the triangle inequality.}, the $E_{\gamma}$ metric. The $E_{\gamma}$ metric was, to our knowledge, originally introduced in \cite{polyanskiy2010channel} to simplify the formula of the DT bound therein. The latter metric has clear operational significance and reduces to the TV in the special case of $\gamma=1$, whereas the former is easier to upperbound. Asymptotically, however, the two metrics behave in the same way. We derive a one-shot upperbound on the first (hence also the second) metric in the resolvability problem. Bounding the new metrics requires more care than the traditional TV to achieve asymptotic tightness.

Particularly interesting is the case where the channel is stationary memoryless and $\gamma$ grows exponentially as the number of channel uses tends to infinity. In this case a single letter formula of the rate of randomness needed to approximate a tensor power output distribution in $E_{\gamma}$ can be obtained from the aforementioned one-shot bound. Here a peculiar feature of approximation in $E_{\gamma}$ emerges: the distribution of each codeword in the generation of the random codebook need not induce the target output distribution through the stationary memoryless channel, and in fact the optimal choice of such a distribution (in the sense of requiring the minimum rate of randomness) generally does \emph{not} induce the target distribution. This is in stark contrast to the case of TV measure, where the codeword distribution \emph{must} induce the target distribution to ensure that the total variation between the output distribution and the target distribution does not converge to its maximum value, 2, asymptotically.

Two applications of the new channel resolvability results are presented. First, the simplest application to lossy source coding yields a new achievability bound on the probability that the distortion lies below a certain number, which in the asymptotic setting recovers the exponent of this probability previously obtained using the method of types (c.f.~\cite{csiszar2011information}). The advantage of the new derivation is its applicability beyond the discrete memoryless framework.

The second application is in the achievability part of wiretap channels, where we propose a novel interpretation of secrecy in terms of the eavesdropper's ability to perform list decoding. In contrast to the previous proofs for wiretap channels using TV-resolvability \cite{hayashi2006general}\cite{bloch2013strong} which only applies when the rate is below the perfect secrecy capacity, the new resolvability in $E_{\gamma}$ yields lower bounds on the required size of the eavesdropper list for all possible rates. This interpretation of security in terms of list size is reminiscent of equivocation \cite{wyner1975wire}, and indeed we obtain the same formula in the asymptotic setting, even though it is not immediate to prove a correspondence between the two. We also consider the case where the eavesdropper wishes to detect that no message is sent with high probability. This is a practical setup because ``no message'' may be a special piece of information which the eavesdropper wants to know with high certainty. We obtain single letter expressions of the tradeoff between the transmission rate, eavesdropper list, and the exponent of the probability that the eavesdropper fails to detect non-message. Those bounds are asymptotically tight for random codes.

\section{Preliminaries}
\subsection{Excess Information Metric}
One natural measure of the discrepancy between two distributions $P$ and $Q$ on the same alphabet may be called the \emph{excess information metric with threshold $\gamma$}:
\begin{align}\label{e_ex}
\mathbb{P}[\imath_{P||Q}(X)>\log\gamma]
\end{align}
where $X\sim P$ and
\begin{align}
\imath_{P||Q}(x):=\log\frac{{\rm d}P}{{\rm d}Q}(x).
\end{align}
Notice that in additional to being more suitable for a one-shot approach, \eqref{e_ex} provides richer information than the relative entropy measure since
\begin{align}
D(P||Q)
&=\int_{[0,+\infty)}\mathbb{P}[\imath_{P||Q}(X)>\tau]{\rm d}\tau
\nonumber
\\
&\quad-\int_{(-\infty,0]}(1-\mathbb{P}[\imath_{P||Q}(X)>\tau]){\rm d}\tau.
\end{align}
We note that the excess information metric does not satisfy a data processing property. More precisely, suppose
$P_X\to P_{Y|X} \to P_{Y}$, $Q_X\to P_{Y|X} \to Q_Y$, then it is \emph{not} always true that
\begin{align}
\mathbb{P}[\imath_{P_X||Q_X}(X)\ge \tau]< \mathbb{P}[\imath_{P_Y||Q_Y}(Y)\ge \tau]
\end{align}
where $(X,Y)\sim P_{XY}$.

\subsection{The $E_{\gamma}(P||Q)$ Metric}
Next we consider another metric which does satisfy the data processing inequality and has a clearer operational meaning.
Given probability distributions $P$, $Q$ and a constant $\gamma\ge1$, define an $f$-divergence \cite{csiszar1967}
\begin{align}
E_{\gamma}(P||Q):=\mathbb{P}[\imath_{P||Q}(X)>\log\gamma]-\gamma\mathbb{P}[\imath_{P||Q}(Y)>\log\gamma]
\end{align}
where $X\sim P$ and $Y\sim Q$.
This quantity was introduced in \cite{polyanskiy2010channel} to simplify the expression of DT bound.
From the Neyman-Pearson lemma we have the alternative formula for the above quantity:
\begin{align}
E_{\gamma}(P||Q)=\max_{A}(P(A)-\gamma Q(A)),
\end{align}
which becomes half of the total variation distance (the $\ell_1$ distance) between $P$ and $Q$ when $\gamma=1$. Some basic properties of $E_{\gamma}$ are in order:
\begin{prop}\label{prop_egamma}
\begin{enumerate}
\item For any event $A$,
\begin{align}
Q(A)\ge \frac{1}{\gamma}(P(A)-E_{\gamma}(P||Q)).
\end{align}
\item  If $P_XP_{Y|X}$ and $Q_XQ_{Y|X}$ are joint distributions on $\mathcal{X}\times\mathcal{Y}$, then
\begin{align}
E_{\gamma}(P_X||Q_X)\le E_{\gamma}(P_XP_{Y|X}||Q_XQ_{Y|X})
\end{align}
where equality holds when $P_{Y|X}=Q_{Y|X}$. In the latter case we obtain the data processing inequality:
\begin{align}
E_{\gamma}(P_Y||Q_Y)\le E_{\gamma}(P_X||Q_X)
\end{align}
\item Given $P_X$, $P_{Y|X}$ and $Q_{Y|X}$,
define
\begin{align}
E_{\gamma}(P_{Y|X}||Q_{Y|X}|P_X):=\mathbb{E}[E_{\gamma}(P_{Y|X}(\cdot|X)||
Q_{Y|X}(\cdot|X))]
\end{align}
where the expectation is w.r.t.~$X\sim P_X$. Then
\begin{align}
E_{\gamma}(P_XP_{Y|X}||P_XQ_{Y|X})= E_{\gamma}(P_{Y|X}||Q_{Y|X}|P_X).
\end{align}
\end{enumerate}
\end{prop}

\section{Achievability Bounds on Excess Information}\label{sec_ach}
We present a one-shot information spectrum achievability bound for resolvability under the excess information metric, which then automatically implies a bound under the $E_{\gamma}$ metric. Consider the setting of Figure~\ref{f2}. The input to the channel $Q_{X|U}$ is equiprobably selected from a codebook $(c_l)_{l=1}^L\in\mathcal{U}^L$. It turns out that codewords are i.i.d.~codewords are usually good enough, and the expected distance from the synthesized distribution $P_{X(\mb{c})}$ to the target distribution $\pi_X$ under the excess information metric is gauged as follows:
\begin{figure}[h!]
  \centering
\begin{tikzpicture}
[node distance=0.5cm,minimum height=10mm,minimum width=10mm,arw/.style={->,>=stealth'}]
  \node[rectangle,rounded corners] (SD) {$l$};
  \node[rectangle,draw,rounded corners] (PX) [right =of SD] {$(c_l)_{l=1}^L$};
  \node[rectangle,draw,rounded corners] (PZX) [right =1cm of PX] {$Q_{X|U}$};
  \node[rectangle,rounded corners] (PZ) [right =1cm of PZX] {$P_X\approx\pi_X$};

  \draw [arw] (SD) to node[midway,above]{} (PX);
  \draw [arw] (PX) to node[midway,above]{} (PZX);
  \draw [arw] (PZX) to node[midway,above]{} (PZ);
\end{tikzpicture}
\caption{Synthesizing a target distribution $\pi_X$ using a random number generator and a codebook $(c_l)_{l=1}^L$.}
\label{f2}
\end{figure}
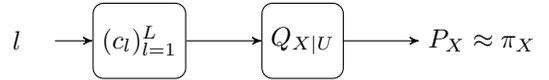

\begin{thm}\label{thm3}
Fix $\pi_X$ and $Q_{UX}=Q_UQ_{X|U}$. Let $\mathbf{c}=[c_1,\dots,c_L]$ be i.i.d.~according to $Q_X$. Define
\begin{align}
P_{X(\mb{c})}:=\frac{1}{L}\sum_{l=1}^L Q_{X|U=c_l}.
\end{align}
Then for any
$\tau,\gamma,\epsilon,\sigma>0$ satisfying $\gamma>\epsilon+\sigma$ and $0<\delta<1$, it holds that
\begin{align}
\mathbb{P}\left[\frac{{\rm d}P_{X(\mb{c})}}{{\rm d}\pi_X}(\hat{X})>\gamma
\right]
&\le \mathbb{P}\left[\frac{{\rm d}Q_X}{{\rm d}\pi_X}(X)>\gamma-\sigma-\epsilon\right]
\nonumber\\
&\quad+\mathbb{P}\left[\frac{{\rm d}Q_{X|U}}{{\rm d}\pi_X}(X|U)> \delta L\sigma\right]
\nonumber\\
&\quad+\frac{\exp(\tau)(\gamma-\sigma-\epsilon)^2}{L(1-\delta)^2\sigma^2}
\nonumber\\
&\quad+\frac{\gamma-\sigma-\epsilon}{\epsilon}\mathbb{P}[\imath_{U;X}(U;X)>\tau]
\label{e_ub}
\end{align}
where conditioned on $\mb{c}$, $\hat{X}\sim P_{X(\mb{c})}$, and $(U,X)\sim Q_U Q_{X|U}$.
\end{thm}
\begin{rem}\label{remweaken}
By setting $\tau\leftarrow-\infty$ and letting $\delta\uparrow 1$, the bound in Theorem~\ref{thm3} can be weakened in the following slightly simpler form:
\begin{align}
\mathbb{P}\left[\frac{{\rm d}P_X}{{\rm d}\pi_X}(\hat{X})>\gamma\right]
&\le \mathbb{P}\left[\frac{{\rm d}Q_X}{{\rm d}\pi_X}(X)>\gamma-\sigma-\epsilon\right]\nonumber%
\\
&\quad+\mathbb{P}\left[\frac{{\rm d}Q_{X|U}}{{\rm d}\pi_X}(X|U)\ge L\sigma\right]\nonumber
\\
&\quad+\frac{\gamma-\sigma-\epsilon}{\epsilon}\label{t1}
\end{align}
The weakened bound \eqref{t1} is still asymptotically tight provided that the exponent with which the threshold $\gamma$ grows is positive; see Corollary~\ref{cor_asymp} below. However, when the exponent is zero (corresponding to the total variation case), we do need $\tau$ in the bound for asymptotic tightness.
\end{rem}
The proof of Theorem~\ref{thm3} is omitted due to space limitations. Next we particularize Theorem~\ref{thm3} to the case of stationary memoryless channels and an exponentially growing threshold $\gamma$, to obtain explicit single-letter formula for the tradeoff between $R$ and the exponent of $\gamma$:
\begin{cor}\label{cor_asymp}
Fix per-letter distributions $\pi_{\sf X}$ and $Q_{\sf UX}=Q_{\sf U}Q_{\sf X|U}$. Let $\mathbf{c}=[c_1,\dots,c_L]$ be i.i.d.~according to $Q^{\otimes n}_{\sf X}$. Define
\begin{align}
P_{{\sf X}^n(\mb{c})}:=\frac{1}{L}\sum_{l=1}^L Q_{{\sf X}^n|{\sf U}^n=c_l}.
\end{align}
Suppose $\gamma=\exp(nE)$ and $L=\exp(nR)$. Then
\begin{align}
\lim_{n\to\infty}\mathbb{E}[E_{\gamma}(P_{{\sf X}^n(\mb{c})}||\pi_{{\sf X}}^{\otimes n})]
&=\lim_{n\to\infty}\mathbb{P}\left[\frac{{\rm d}P_{{\sf X}^n(\mb{c})}}{{\rm d}\pi_{\sf X}^{\otimes n}}(\hat{\sf X}^n)>\gamma\right]
\nonumber
\\
&=0
\end{align}
provided that
\begin{align}\label{e_e}
E>D(Q_{\sf X}||\pi_{\sf X})+[I(Q_{\sf U},Q_{\sf X|U})-R]^+,
\end{align}
where conditioned on $\mb{c}$, the vector $\hat{\sf X}^n\sim P_{{\sf X}^n(\mb{c})}$.
Moreover, the bound in \eqref{e_e} is tight.
\end{cor}
\begin{proof}[Proof of Achievability]Choose $E'$ such that
\begin{align}
E>E'>D(Q_{\sf X}||\pi_{\sf X})+[I(Q_{\sf U},Q_{\sf X|U})-R]^+.
\end{align}
Set $\delta=\frac{1}{2}$, $\gamma=\exp(nE)$, $L=\exp(nR)$, $\epsilon=\exp(nE)-\exp(nE')$ and $\sigma=\frac{1}{2}(\gamma-\epsilon)=\frac{1}{2}\exp(nE')$, and apply \eqref{t1}. Notice that
\begin{align}
\mathbb{E}\left[\frac{{\rm d}Q_{X|U}}{{\rm d}\pi_X}(X|U)\right]=n[I(Q_{\sf U},Q_{\sf X|U})+D(Q_{\sf X}||\pi_{\sf X})]
\end{align}
where $(X,U)\sim Q_{\sf XU}^{\otimes n}$.
By the law of large numbers, the first and second terms in \eqref{t1} vanish because
\begin{align}
D(Q_{\sf X}||\pi_{\sf X})&<E';
\\
I(Q_{\sf U},Q_{\sf X|U})+D(Q_{\sf X}||\pi_{\sf X})&<E'+R
\end{align}
are satisfied.
\end{proof}
The basic idea for the proof of the tightness of \eqref{e_e} (the converse) is as follows:
given a codebook $\mb{c}$ define $\mathcal{A}:=\bigcup_{l=1}^LT_{Q_{\sf X|U},\delta}(c_l)$, where $T_{Q_{\sf X|U},\delta}(c_l)$ denotes the $Q_{\sf X|U}$-typical sequences given $c_l$. Then it can be shown that when $E$ is less than the right hand side of \eqref{e_e}, it holds that $(P_{{\sf X}^n(\mb{c})}-\gamma Q_{\sf X}^{\otimes n})(\mathcal{A})\to 1$ for some $\delta>0$.

\section{Application to Lossy Source Coding}\label{seclikelihood}
The simplest application of the new resolvability result is to derive a one-shot achievability bound for source coding, which is most fitting in the regime of low rate and exponentially decreasing success probability. The method is applicable to general sources. In the special case of i.i.d.~sources, it recovers the ``success exponent'' in lossy source coding originally derived by the method of types \cite{csiszar1967} for discrete memoryless sources.

\begin{thm}\label{thm_source}
Consider a source with distribution $\pi_X$ and a distortion function $d(\cdot,\cdot)$ on $\mathcal{U}\times\mathcal{X}$. For any distribution $Q_UQ_{X|U}$, $\gamma\ge1$, $d>0$ and integer $L$, there exists a stochastic encoder $\pi_{U|X}$ such that the size of the support of $\pi_U$ is at most $L$ and
\begin{align}
\mathbb{P}[d(\bar{U},\bar{X})\le d]\ge\frac{1}{\gamma}\left(\mathbb{P}[d(U,X)\le d]-\varepsilon\right)\label{e21}
\end{align}
where $(\bar{U},\bar{X})\sim \pi_{UX}$, $(U,X)\sim Q_{UX}$, and $\varepsilon$ is an upper-bound on the right hand side of \eqref{e_ub}.
\end{thm}
\begin{proof}
Given a codebook $(c_1,\dots,c_L)\in\mathcal{U}$,
let $P_U$ be the equiprobable distribution on $(c_1,\dots,c_L)$ and set
\begin{align}
P_{UX}:=Q_{X|U}P_U.
\end{align}
The likelihood encoder is then defined as a random transformation
\begin{align}
\pi_{U|X}:=P_{U|X}
\end{align}
so that the joint distribution of the codeword selected and the source realization $X$ is
\begin{align}
\pi_{UX}=\pi_XP_{U|X}
\end{align}
From Proposition~\ref{prop_egamma} we obtain
\begin{align}
&\quad\gamma\mathbb{P}[d(\bar{U},\bar{X})\le d]
\nonumber
\\
&\ge \mathbb{P}[d(\hat{U},\hat{X})\le d]\nonumber
-E_{\gamma}(P_{XU}||\pi_{XU})
\\
&= \mathbb{P}[d(\hat{U},\hat{X})\le d]
-E_{\gamma}(P_X||\pi_X)
\end{align}
where $(\hat{U},\hat{X})\sim P_{UX}$,
which yields
\begin{align}
&\quad\gamma\mathbb{E}_{\mb{c}}\mathbb{P}[d(\bar{U},\bar{X})\le d]
\nonumber
\\
&\ge \mathbb{P}[d(U,X)\le d]
-\mathbb{E}_{\mb{c}}E_{\gamma}(P_{X}||\pi_{X})\label{e_2s}
\end{align}
where in \eqref{e_2s} we used the fact that $\mathbb{E}_{\mb{c}} P_{UX}=Q_{UX}$.
Finally we can choose a codebook such that $\mathbb{P}[d(\bar{U},\bar{X})\le d]$ is at least its expectation.
\end{proof}

\begin{rem}
In the i.i.d.~setting,
let $R(\pi_{\sf X},d)$ be the rate-distortion function when the source has per-letter distribution $\pi_{\sf X}$. The distortion function for the block is derived from the per-letter distortion by
\begin{align}
d^{(n)}(u^n,x^n):=\frac{1}{n}\sum_{i=1}^n d(u_i,x_i).
\end{align}
Let $(\bar{\sf X}^n,\bar{\sf U}^n)$ be the source-reconstruction pair distributed according to $\pi_{{\sf X}^n{\sf U}^n}$.
If $0\le R<R(\pi_{\sf X},d)$, the maximal probability that the distortion does not exceed $d$ converges to zero with the exponent
\begin{align}
\lim_{n\to\infty}\frac{1}{n}\log\frac{1}{\mathbb{P}[d^{(n)}(\bar{\sf U}^n,\bar{\sf X}^n)\le d]}=G(R,d)
\end{align}
where
\begin{align}\label{e_sexp}
G(R,d):=\min_Q[D(Q||P)+[R(Q,d)-R]^+].
\end{align}
A weaker achievability result than \eqref{e_sexp} was proved in \cite[p168]{omura1975lower}, whereas the final form \eqref{e_sexp} is given in \cite[p158, Ex6]{csiszar2011information} based on method of types. Here we can easily prove the achievability part of \eqref{e_sexp} using Theorem~\ref{thm_source} and Corollary~\ref{cor_asymp} by setting $Q_{\sf X}$ to be the minimizer of \eqref{e_sexp} and $Q_{\sf U|X}$ to be such that
\begin{align}
\mathbb{E}d({\sf U,X})&\le d,
\\
I(Q_{\sf X},Q_{\sf X|U})&\le R.
\end{align}
Then $\gamma_n=\exp(nE)$ with
\begin{align}\label{esuccess}
E>D(Q_X||\pi_X)+[I(U;X)_Q-R]^+,
\end{align}
ensures that
\begin{align}
\mathbb{P}[d^{(n)}(\bar{\sf U}^n,\bar{\sf X}^n)\le d]\ge\frac{1}{2}\exp(-nE)
\end{align}
for $n$ large enough, by the law of large numbers.
\end{rem}
\begin{rem}
Since the $E_{\gamma}$ metric reduces to TV when $\gamma=1$, Theorem~\ref{thm_source} generalizes the likelihood source encoder based on the standard soft-covering/resolvability lemma \cite{song}. In \cite{song}, the error exponent for the likelihood source encoder at rates \emph{above} the rate-distortion function is analyzed using the exponential decay of TV in the approximation of output statistics, and the exponent does not match the optimal exponent in \cite{csiszar1967}. It is also possible to upperbound the success exponent of the TV-based likelihood encoder at rates \emph{below} the rate-distortion function by analyzing the exponential convergence to $2$ of TV in the approximation of output statistics; however that does not yield the optimal exponent \eqref{e_sexp} either. The power of $E_{\gamma}$-resolvability lies in the ability to \emph{convert a large deviation analysis into an excercise of the law of large numbers,} that is, we only care about whether $E_{\gamma}$ converges to $0$, but not the speed, even when dealing with error exponent problems.
\end{rem}

\section{Application to Wiretap Channels}\label{sec_wiretap}
Next we apply the $E_{\gamma}$-resolvability to the wiretap channel $P_{YZ|X}$ as depicted in Figure~\ref{fwiretap}. The receiver and the eavesdropper observe $y\in\mathcal{Y}$ and $z\in\mathcal{Z}$, respectively. Given a codebook $(c_{wl})$, the input to the channel is $c_{wl}$ where $w\in\{1,\dots,M\}$ is the message to be sent and $l$ is equiprobably chosen from $\{1,\dots,L\}$ to randomize the eavesdropper's observation. Moreover, the eavesdropper's observation has the distribution $\pi_Z$ when no message is sent. For general wiretap channels the performance may be enhanced by appending a conditioning channel $Q_{X|U}$ at the output of the encoder \cite{hayashi2006general}. But in that case the same analysis can be carried out for the new wiretap channel $Q_{YZ|U}$. Thus the model in Figure~\ref{fwiretap} entails no loss of generality.
\begin{figure}[h!]
  \centering
\begin{tikzpicture}
[node distance=0.5cm,minimum height=10mm,minimum width=15mm,arw/.style={->,>=stealth'}]
  \node[rectangle,draw] (T) {$P_{YZ|X}$};
  \node[rectangle,draw,rounded corners] (A) [left =of T] {$(c_{wl})$};
  \node[rectangle,draw,rounded corners] (B) [right =of T] {Receiver};
  \node[rectangle,draw,rounded corners] (E) [above =of T] {Eavesdropper};
  \node[rectangle] (M) [left =of A] {};
  \node[rectangle] (Mhat) [right =of B] {};
  \node[rectangle] (L) [above =of A] {};

  \draw [arw] (A) to node[midway,above]{} (T);
  \draw [arw] (T) to node[midway,above]{} (B);
  \draw [arw] (T) to node[midway,above]{} (E);
  \draw [arw] (M) to node[midway,left]{$w$} (A);
  \draw [arw] (B) to node[midway,right]{$\hat{w}$} (Mhat);
  \draw [arw] (L) to node[midway,above]{$l$} (A);
\end{tikzpicture}
\caption{The wiretap channel}
\label{fwiretap}
\end{figure}
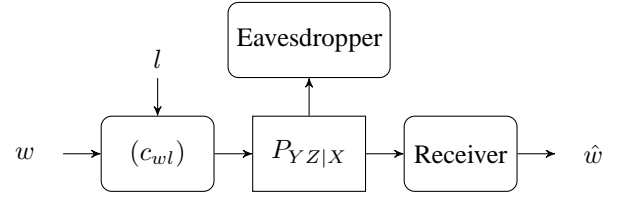
We need the following definitions to quantify the eavesdropper's knowledge.
\begin{defn}
For a fixed codebook we say the eavesdropper can perform $(A,T,\bar{\epsilon})$-decoding if when no message is sent, it detects no message with probability at least $1-A^{-1}$; and when a message $m$ is sent, it can produce a list of $T$ messages containing $m$ with probability at least $\epsilon_m$ such that
\begin{align}
\bar{\epsilon}=\frac{1}{M}\sum_{m=1}^M\epsilon_m.
\end{align}
\end{defn}
For stationary memoryless channels, the quantities $P_{ZY|X}$, $M$ and $L$ in Figure~\ref{fwiretap} are identified as $P_{{\sf Z}^n{\sf Y}^n|{\sf X}^n}$, $\exp(nR)$ and $\exp(nR_{\sf L})$.

We consider an \emph{$(M,L,Q_X)$-random code}, which is defined as the ensemble of the codebook $(c_{wl})$, $w\in\{1,\dots,M\}$, $l\in\{1,\dots,L\}$ where each codeword is i.i.d.~chosen according to $Q_X$.
The following definition captures the asymptotic performance of the eavesdropper:
\begin{defn}\label{def_ach}
Fix $(R,R_{\sf L})$. The rate pair $(\alpha,\tau)$ is $\bar{\epsilon}$-achievable by the eavesdropper if there exist sequences $\{A_n\}$ and $\{T_n\}$ with
\begin{align}
\lim_{n\to\infty}\frac{1}{n}\log T_n&=\tau
\\
\lim_{n\to\infty}\frac{1}{n}\log A_n&=\alpha
\end{align}
such that for sufficiently large $n$, the eavesdropper can achieve $(A_n,T_n,\bar{\epsilon})$-decoding with high probability when the codebook is the $(\exp(nR),\exp(nR_{\sf L}),Q_{\sf X}^{\otimes n})$-random code.
\end{defn}
Then we have the following result:
\begin{thm}\label{region}
For any $Q_{\sf X}$, $R$, $R_{\sf L}$ and $0<\bar{\epsilon}<1$,
the pair $(\alpha,\tau)$ is $\bar{\epsilon}$-achievable by the eavesdropper in the sense of Definition~\ref{def_ach} iff
\begin{align}\label{e_reg}
\left\{
\begin{array}{l}
  \alpha\le D(Q_{\sf Z}||\pi_{\sf Z})+[I(Q_{\sf X},P_{\sf Z|X})-R-R_{\sf L}]^+ \\
  \tau \ge R-[I(Q_{\sf X},P_{\sf Z|X})-R_{\sf L}]^+
  \end{array}
\right.
\end{align}
where $Q_{\sf X}\to P_{\sf Z|X}\to Q_{\sf Z}$.
\end{thm}
\begin{rem}
From the noisy channel coding theorem, the supremum randomization rate $R_{\sf L}$ such that the sender can reliably transmit messages at the rate $R$ is $I(Q_{\sf X},P_{\sf Y|X})-R$. The larger $R_{\sf L}$ the less reliably the eavesdropper can decode, so the optimal encoder chooses $R_{\sf L}$ as close to this supremum as possible. Thus Theorem~\ref{region} implies that to reliably transmit messages at the rate $R$, codebooks can be selected such that the eavesdropper cannot perform $(\exp(n\alpha),\exp(n\tau),\bar{\epsilon})$ for large $n$ if there exists some $Q_{\sf X}$ such that
\begin{align}\label{e_reg1}
\alpha>D(Q_{\sf Z}||\pi_{\sf Z})+[I(Q_{\sf X},P_{\sf Z|X})-I(Q_{\sf X},P_{\sf Y|X})]^+
\end{align}
or
\begin{align}\label{e_reg2}
\tau<R-[I(Q_{\sf X},P_{\sf Z|X})-I(Q_{\sf X},P_{\sf Y|X})+R]^+.
\end{align}
\end{rem}
\begin{rem}
In general the sender-receiver want to minimize $\alpha$ and maximize $\tau$ obeying the tradeoff \eqref{e_reg1}, \eqref{e_reg2} by selecting $Q_{\sf X}$. In the special case where $\alpha$ has no importance and $R$ is larger than the secrecy capacity $C:=\sup_{Q_{\sf X}}\{I(Q_{\sf X},P_{\sf Y|X})-I(Q_{\sf X},P_{\sf Z|X})\}$, we see from \eqref{e_reg2} that the supremum $\tau$ is $C$. The formula is the same as the equivocation measure defined as $\frac{1}{n}H(W|{\sf Z}^n)$ \cite{wyner1975wire}, but technically our result does not follow directly from the lower bound on equivocation, since it may be possible that the a posterior distribution of $W$ is concentrated on a small list but has a tail spread over an exponentially large set, resulting a large equivocation.
\end{rem}
The (eavesdropper) achievability part of Theorem~\ref{region} follows by analyzing the eavesdropper decoding ability for different cases of the rates $(R,R_{\sf L})$. The (eavesdropper) converse part of Theorem~\ref{region} follows by applying the following non-asymptotic bounds to different cases of $(R,R_{\sf L})$ and invoking Corollary~\ref{cor_asymp}.
\begin{thm}\label{thmlem}
In the wiretap channel, fix an arbitrary distribution $\mu_Z$ and a measurable subset $\mathcal{D}_0\subseteq\mathcal{Z}$.
Suppose the eavesdropper can either detect that no message is sent upon observing $z\in\mathcal{D}_0$ with
\begin{align}
\mu_Z(\mathcal{D}_0)\ge 1-A^{-1}
\end{align}
or outputs a list of $T(z)$ messages upon observing $z\notin \mathcal{D}_0$ that contains the actual message $m\in\{1,\dots,M\}$ with probability at least $1-\epsilon_m$. Define the average quantities
\begin{align}
T&:=\frac{1}{\mu_Z(\mathcal{D}_0^c)}\int_{\mathcal{D}_0^c} T(z){\rm d}\mu_Z(z),
\\
\bar{\epsilon}&:=\frac{1}{M}\sum_{m=1}^M \epsilon_m.
\end{align}
Then,
\begin{align}\label{e_45}
\frac{1}{A}\ge\frac{1}{\gamma}\left(1-\bar{\epsilon}-E_{\gamma}(P_Z||\pi_Z)\right),
\end{align}
where we recall that $\pi_Z$ is the non-message distribution, and
\begin{align}\label{e_46}
\frac{T}{MA}\ge\frac{1}{\gamma}\left(1-\bar{\epsilon}-\frac{1}{M}\sum_{m=1}^M E_{\gamma}(P_{Z|W=m}||\mu_Z)\right).
\end{align}
\end{thm}
From the eavesdropper viewpoint, a larger $A$ and a smaller $T$ is more desirable since it will then be able to find out that no message is sent with smaller error probability or narrow down to a smaller list when a message is sent. This observation agrees with \eqref{e_45} and \eqref{e_46}: a smaller $\gamma$ implies a higher degree of approximation, and hence higher indistinguishability of output distributions which is to the eavesdropper disadvantage.

\section{Discussion}
As we have demonstrated, the achievability part of resolvability in $E_\gamma$ has various applications in information theory, especially for bounding rare event probabilities. (c.f.~\eqref{e21}\eqref{e_45} and \eqref{e_46}). However the asymmetry of $E_{\gamma}$ (when $\gamma>1$) places a limitation on $E_{\gamma}$-resolvability in certain problems. In particular, there is no counterpart of Theorem~\ref{thm3} for $E_{\gamma}(\pi_X||P_X)$.

\section*{Acknowledgment}
Our initial focus was on the excess information metric for resolvability, as in Theorem~\ref{thm3}. We gratefully acknowledge Yury Polyanskiy for bringing the $E_{\gamma}$ metric to our sight and showing us the useful properties of this metric.
This work was supported by NSF under Grants CCF-1350595, CCF-1116013, CCF-1319299, CCF-1319304, and the Air Force Office of Scientific Research under Grant FA9550-12-1-0196.

\bibliographystyle{ieeetr}
\bibliography{SM}

\begin{thebibliography}{10}

\bibitem{wyner1975common}
A.~D. Wyner, ``The common information of two dependent random variables,'' {\em
  IEEE Transactions on Information Theory}, vol.~21, no.~2, pp.~163--179, 1975.

\bibitem{han1993approximation}
T.~Han and S.~Verd\'{u}, ``Approximation theory of output statistics,'' {\em
  IEEE Transactions on Information Theory}, vol.~39, no.~3, pp.~752--772, 1993.

\bibitem{han2013reliability}
T.~S. Han, H.~Endo, and M.~Sasaki, ``Reliability and security functions of the
  wiretap channel under cost constraint,'' {\em IEEE Transactions on
  Information Theory}, vol.~60, no.~11, pp.~6819\--6843, 2014.

\bibitem{cuff2012distributed}
P.~Cuff, ``Distributed channel synthesis,'' {\em IEEE Transactions on
  Information Theory}, vol.~59, pp.~7071--7096, Nov. 2013.

\bibitem{csiszar1996}
I.~Csisz\'{a}r, ``Almost independence and secrecy capacity,'' {\em Problems
  Inf. Transmission}, vol.~32, no.~1, pp.~40\--47, 1996.

\bibitem{hayashi2006general}
M.~Hayashi, ``General nonasymptotic and asymptotic formulas in channel
  resolvability and identification capacity and their application to the
  wiretap channel,'' {\em IEEE Transactions on Information Theory}, vol.~52,
  pp.~1562--1575, Apr. 2006.

\bibitem{bloch2013strong}
M.~Bloch and N.~Laneman, ``Strong secrecy from channel resolvability,'' {\em
  IEEE Transactions on Information Theory}, vol.~59, pp.~8077--8098, Dec. 2013.

\bibitem{song}
E.~C. Song, P.~Cuff, and H.~V. Poor, ``The likelihood encoder for lossy source
  compression,'' {\em arXiv:1408.4522}, Aug.~2014.

\bibitem{steinberg1996simulation}
Y.~Steinberg and S.~Verd\'{u}, ``Simulation of random processes and
  rate-distortion theory,'' {\em IEEE Transactions on Information Theory},
  vol.~42, no.~1, pp.~63--86, 1996.

\bibitem{polyanskiy2010channel}
Y.~Polyanskiy, H.~V. Poor, and S.~Verd{\'u}, ``Channel coding rate in the
  finite blocklength regime,'' {\em IEEE Transactions on Information Theory},
  vol.~56, no.~5, pp.~2307--2359, 2010.

\bibitem{csiszar2011information}
I.~Csiszar and J.~K{\"o}rner, {\em Information theory: coding theorems for
  discrete memoryless systems}.
\newblock Cambridge University Press, 2011.

\bibitem{wyner1975wire}
A.~D. Wyner, ``The wire-tap channel,'' {\em The Bell System Technical Journal},
  vol.~54, no.~8, pp.~1355--1387, 1975.

\bibitem{csiszar1967}
I.~Csisz\'{a}r, ``Information-type measures of difference of probability
  distributions and indirect observation,'' {\em Studia Sci. Math. Hungar.},
  vol.~2, pp.~229\--318, 1967.

\bibitem{omura1975lower}
J.~K. Omura, ``A lower bounding method for channel and source coding
  probabilities,'' {\em Information and Control}, vol.~27, no.~2, pp.~148--177,
  1975.

\end{thebibliography}
\end{document}